\documentclass[a4paper]{article}

\usepackage{amsthm}
\usepackage{amssymb}
\usepackage[left=35mm,right=35mm,top=40mm,bottom=45mm]{geometry}
\usepackage{url}

\theoremstyle{plain}
\newtheorem{theorem}{Theorem}[section]
\newtheorem{proposition}[theorem]{Proposition}

\newtheorem{corollary}[theorem]{Corollary}
\theoremstyle{definition}

\newtheorem{example}{Example}
\newtheorem*{conjecture}{Conjecture}

\theoremstyle{remark}

\newcommand{\Suff}{\textit{Suff}}
\newcommand{\Pref}{\textit{Pref}}
\newcommand{\Fact}{\textit{Fact}}

\renewcommand{\epsilon}{\varepsilon}

\begin{document}

\date{\today}

\title{On the Minimal Uncompletable Word Problem}

\author{
   \textsc{Gabriele Fici}
   \thanks{I3S, CNRS and Universit\'{e} de Nice-Sophia Antipolis, France.
           Mail: {\tt fici@i3s.unice.fr}}
              \and{\textsc{Elena V. Pribavkina}
   \thanks{Ural State University, 620083 Ekaterinburg, Russia.
           Mail: {\tt elena.pribavkina@usu.ru}}
   }
   \and{\textsc{Jacques Sakarovitch}
   \thanks{LTCI, CNRS / Telecom ParisTech, France.
           Mail: {\tt sakarovitch@enst.fr}}
   }
}

\maketitle

\begin{abstract}
Let $S$ be a finite set of words such that $\Fact(S^*)\ne\Sigma^*$. We deal with the problem of finding bounds on the minimal length of words in $\Sigma^{*}\setminus \Fact(S^{*})$ in terms of the maximal length of words in $S$.
\end{abstract}


\section{Introduction}\label{sec:intro}

A finite set $S$ of (finite) words over an alphabet $\Sigma$ is said to be \emph{complete} if $\Fact(S^{*})$, the set of factors of $S^{*}$, is equal to $\Sigma^{*}$, that is, if every word of $\Sigma^{*}$ is a factor of, or can be completed by multiplication on the left and on the right as, a word of $S^{*}$.

If $S$ is not complete, $\Sigma^{*}\setminus \Fact(S^{*})$ is not empty and a word in this set of minimal length is called a \emph{minimal uncompletable word} (with respect to the non-complete set $S$).

The problem of finding minimal uncompletable words and their length was introduced by Restivo  \cite{Restivo81}, who conjectured that there is a quadratic upper bound for the length of a minimal uncompletable word for $S$ in terms of the maximal length of words in $S$.

A more general related question of deciding whether a given regular language $L$ satisfies one of the properties $\Sigma^*=\Fact(L)$, $\Sigma^*=\Pref(L)$, $\Sigma^*=\Suff(L)$ has been recently considered by Rampersad et al.\ in \cite{Shall09}, where the computational complexity of the aforesaid problems in case $L$ is represented by a DFA or NFA is studied. In the particular case $L=S^{*}$ for $S$ being a finite set of words -- which is the case that is of interest for us -- the authors mention that the complexity of deciding whether or not $\Sigma^*=\Fact(S^{*})$ is still an open problem.

In this note, we show by mean of an example that the length of a minimal uncompletable word for a set $S$ whose longest word is of length $k$ seems to grow as $3k^{2}$ asymptotically and at least gets larger than $2k^{2}$ for effectively computed values, thus improving on a previous example given by Antonio Restivo \cite{Restivo81}.
The computations of a minimal uncompletable word for the successive values of $k$ in the parametrized example were made on the {\sc Vaucanson} platform for computing automata \cite{vaucanson}. This result is briefly mentioned in \cite{BePeReu09}.

The previous attempts to studying non-complete sets of words lead us to the following formulation.

Let $S\subseteq \Sigma^{*}$ and denote by

\begin{eqnarray*}
uwl(S)=\left\{ \begin{array}{ll}
\min\{|x| \mbox{ : } x\in \Sigma^{*}\setminus \Fact(S^{*})\} & \mbox{if $\Sigma^{*}\setminus \Fact(S^{*})\neq 0$},\\
0 & \mbox{otherwise}
\end{array} \right.\\
\end{eqnarray*}

and by

$$UWL(k,\sigma)=\max\{uwl(S) \mbox{ : } S\subseteq \Sigma^{\leq k}, |\Sigma|=\sigma\}$$

In fact we shall be interested by the case of binary alphabet, and we write $UWL(k)=UWL(k,2)$. The problem is to find upper and lower bounds for $UWL(k)$.

\section{Bounds on the length of minimal uncompletable words}\label{sec:uncovered}

\begin{proposition}\cite{Restivo81}
\label{Restivo}
Let $k$ be an integer and let $S$ be a finite set of words whose maximal length is $k$ and such that there exists a word $u$ of length $k$ with the property that no element of $S$ is a factor of $u$. Then $S$ is non-complete and the word

$$w=(ua)^{k-1}u\hspace{10mm}\mbox{with $a\in \Sigma$}$$

is an uncompletable word for $S$.
\end{proposition}

A direct consequence of this statement is then

\begin{corollary}\cite{Restivo81}
For any integer $k\geq 2$ and any word $u$ in $\Sigma^{k}$, the set $S=\Sigma^{k}\setminus \{u\}$ is non-complete.
\end{corollary}

Actually, if $S=\Sigma^k\setminus\{u\}$ and $u$ is an unbordered word, it can be proved that the uncompletable word from Proposition \ref{Restivo} is also the shortest such word:

\begin{proposition}\cite{Prib09}
For any integer $k\geq 2$ and any unbordered word $u\in\Sigma^k$, a shortest uncompletable word of $S=\Sigma^{*}\setminus \{u\}$ has length $k^2+k-1$.
\end{proposition}

\begin{corollary} For any $k\ge2$ we have $UWL(k)\ge k^2+k-1.$
\end{corollary}

Of course, if $S$ contained in $\Sigma^{*}$ is non-complete and if $S\cup T$ is also contained in $\Sigma^{*}$ and non-complete, any uncompletable word for $S\cup T$ is uncompletable for $S$ and $uwl(S\cup T)\geq uwl(S)$.

The ``game'' is thus to start from a set $S$ of the form $\Sigma^{*}\setminus \{u\}$ and to find a subset $T$ of words of length shorter than $k$ such that $S\cup T$ remains non-complete and the length of minimal uncompletable words increases as much as possible. This is the way that the bound $k^{2}+k-1$ was already improved in \cite{Restivo81}:

\begin{example}
Let $k=4$ and let
$$S_4=\Sigma^4\setminus \{aabb\}\cup \{ab,ba,aba,baa,bab,bba\}$$
Then
$$w=(aabb)aaa(aabb)baa(aabb)bbb(aabb)$$
is a minimal uncompletable word for $S_4$.
\end{example}

Note that in this example the shortest uncompletable word maintains the structure $uv_1uv_2\cdots uv_{k-1}u$ of the uncompletable word from Proposition \ref{Restivo}, but the intermediate words $v_i$ this time have length $k-1$. This example led Restivo to conjecture that $UWL(k)\le2k^{2}$. More precisely:

\begin{conjecture}\cite{Restivo81}
If $S$ is a non-complete set and $k$ is the maximal length of words in $S$, there exists an uncompletable word of length at most $2k^2$. Moreover this word is of the form $uv_1uv_2u\cdots uv_{k-1}u$, where $u$ is the suitable word of length $k$ and $v_1,v_2,\ldots,v_{k-1}$ are words of length less than or equal to $k$.
\end{conjecture}

\begin{example}\label{genrest}
Let $k>4$ and let
$$S_k=\Sigma^k\setminus \{a^{k-2}bb\} \cup \Sigma ba^{k-4}\Sigma \cup \Sigma ba \cup J_{k}$$
where $J_{k}=\bigcup_{i=1}^{k-3}(ba^{i}\Sigma \cup a^{i}b)$. We computed that for $5\leq k\leq 12$ the word
$$w=(a^{k-2}bb) a^{k-1} (a^{k-2}bb) ba^{k-2} ((a^{k-2}bb) b^{2}a^{k-3})^{k-3} (a^{k-2}bb)$$
is a minimal uncompletable word for $S_k$. Thus $UWL(k)\geq 2k^{2}-2k+1$ for $5\le k\le12$.
Using a similar technique as in \cite{Prib09}, it can be proved that this word is uncompletable for each $k\geq 5$, but we are not aware whether this word is minimal uncompletable for $k>12$.
\end{example}

Unfortunately, it is not true in general that $UWL(k)\le2k^{2}$. Indeed, we have

\begin{example}\label{ex:contrex}
Let $k>6$ and let

$$S'_k=\Sigma^k\setminus \{a^{k-2}bb\} \cup \Sigma ba^{k-4}\Sigma \cup \Sigma ba \cup b^{4}\cup J_{k}$$

where $J_{k}=\bigcup_{i=1}^{k-3}(ba^{i}\Sigma \cup a^{i}b)$. We computed that, for $7\leq k\leq 12$,

\begin{eqnarray*}
w &=& (a^{k-2}bb)a^{k-1} (a^{k-2}bb) ba^{k-4} ((a^{k-2}bb)ba(a^{k-2}bb)bba^{k-5})^{k-6}\\
&&(a^{k-2}bb)ab (a^{k-2}bb) bba^{k-3} (a^{k-2}bb) ba^{k-3} (a^{k-2}bb)
\end{eqnarray*}

is a minimal uncompletable word for $S'_k$. Thus $UWL(k)\geq 3k^2-9k+1$ for $7\le k\le 12$.
\end{example}

The set $S'_k$ is obtained from the set $S_k$ considered in Example \ref{genrest} by adding just the word $b^4$.

\section{On the structure of minimal uncompletable words}

Let $u$ be an unbordered word of length $k$, and $S=\Sigma^{k}\setminus \{u\}$. Any uncompletable word for $S$ must contain the word $u$ as a factor, and any word that contains an unbordered factor $u$ can be uniquely written under the form

$$w=v_0uv_1uv_{2}\cdots v_muv_{m+1}$$

with $v_i\in \Sigma^*\setminus \Sigma^*u\Sigma^*$.

Actually, we can say a little bit more on the structure of minimal uncompletable words.

\begin{proposition}\label{prop:struct}
Let $u$ be an unbordered word of length $k$ and $S=\Sigma^{k}\setminus
\{u\}$. Then $u$ is both a prefix and a suffix of any minimal
uncompletable word for $S$, that is, any minimal uncompletable word for
$S$ is of the form

$$w=uv_1uv_{2}u\cdots v_mu$$

with $v_i\in \Sigma^*\setminus \Sigma^*u\Sigma^*$, $m\ge1$.

\end{proposition}

\begin{proof}
Let $w$ be any minimal uncompletable word for $S$. Arguing by contradiction, suppose that $u$ is not a suffix of $w$.
Let $w=w'x$, with $x\in \Sigma$. By the minimality of $w$, we have
$w'\in \Fact(S^{*})$, i.e.\ $w'$ can be covered by words in $S$. Since
$S$ does not contain words longer than $k$, there must exist a prefix
$p$ of $w'$ such that $p\in \Suff(S^{*})$ and $|p|>|w'|-k$, i.e.\
$|p|\geq |w|-k$. But then $w$ could be written as $w=pz$, with $|z|\leq
k$ and $z\ne u$. This implies that $w$ could be covered by words of $S$, which is a contradiction.

In an analogous way one can prove that $u$ must be a prefix of $w$.

\end{proof}

Note that Proposition \ref{prop:struct} still holds for non-complete sets of the form $S=\Sigma^{k}\setminus \{u\}\cup T$, for $u$ an unbordered word of length $k$ and $T$ a set of words of length shorter than $k$.

What about the lengths of factors $v_{i}$'s? In all the examples above each $v_{i}$ has length shorter than $k$. Nevertheless, minimal uncompletable words for which this property is no longer true exist.

\begin{example}
Let
$$S_5=\Sigma^k\setminus \{a^3bb\} \cup \Sigma ba\Sigma \cup \Sigma ba \cup J_{5}$$
where $J_{5}=\bigcup_{i=1}^{2}(ba^{i}\Sigma \cup a^{i}b)$, the set as in the Example \ref{genrest} for $k=5$. Then 

$$w=(a^{3}bb)aaaa(a^{3}bb)baa(a^{3}bb)bbabaa(a^{3}bb)baa(a^{3}bb)$$

is a minimal uncompletable word for $S_{5}$.
\end{example}

\section*{Acknowledgements}
A part of this research was done in 2006 during the visits of second- and third-named authors to the University of Salerno and was supported by the MIUR project ``Formal languages and Automata: Methods, Models and Applications''. The authors are grateful to Jeffrey Shallit for fruitful discussions on the subject at the conference WORDS 2009 and encouraging them to write this note.

\bibliographystyle{plain}

\bibliography{incompletable}
\end{document}